\documentclass[reqno, eucal]{amsart}

\usepackage{epic,eepic}
\usepackage[mathscr]{eucal} 
\usepackage{psfrag}     

\usepackage{amsmath}
\usepackage{amsthm}

\usepackage[dvips]{graphicx,color}  

\usepackage{amssymb}
\usepackage{epic,eepic}
\usepackage{amscd}
\usepackage{longtable}
\usepackage{array}

\newtheorem{theorem}{Theorem}
\newtheorem{lemma}[theorem]{Lemma}

\newtheorem{proposition}[theorem]{Proposition}

\theoremstyle{definition}

\newcommand{\Z}{{\mathbb Z}}

\newcommand{\C}{{\mathbb C}}

\begin{document}

\title[$A^{(1)}_n$ face model]
{Remarks on  
$\boldsymbol{A^{(1)}_n}$ face weights}
\author{Atsuo Kuniba}
\address{Institute of Physics, University of Tokyo, 
Komaba, Tokyo 153-8902, Japan}
\maketitle

\begin{center}{\bf Abstract}\end{center}

Elementary proofs are presented for the factorization of the elliptic 
Boltzmann weights of the $A^{(1)}_n$ face model, and for the sum-to-1 property
in the trigonometric limit, at a special point of the spectral parameter.
They generalize recent results obtained in the context of 
the corresponding trigonometric vertex model.

\vspace{0.2cm}

\section{Introduction}
In the recent work \cite{KMMO}, 
the quantum $R$ matrix for the symmetric tensor representation of the 
Drinfeld-Jimbo quantum affine algebra $U_q(A^{(1)}_n)$ was revisited. 
A new factorized formula at 
a special value of the spectral parameter and 
a certain sum rule called sum-to-1 were established.
These properties have led to vertex models that can be interpreted as 
integrable Markov processes on one-dimensional lattice 
including several examples studied earlier \cite[Fig.1,2]{Kuan}.
In this note we report analogous properties of the 
Boltzmann weights for yet another class of 
solvable lattice models known as 
IRF (interaction round face) models \cite{Bax} or face models for short.
More specifically, we consider the elliptic fusion
$A^{(1)}_n$ face model corresponding to the 
symmetric tensor representation \cite{JMO,JKMO}.
For $n=1$, it reduces to \cite{ABF} and \cite{DJKMO} 
when the fusion degree is 1 and general, respectively.
There are restricted and unrestricted versions of the model.
The trigonometric case of the latter reduces to the $U_q(A^{(1)}_n)$ vertex model
when the site variables tend to infinity. See Proposition \ref{pr:knh}. 
In this sense Theorem \ref{th:s} and Theorem \ref{th:s1} given below, 
which are concerned with the unrestricted version,  
provide generalizations of \cite[Th.2]{KMMO} and \cite[eq.(30)]{KMMO}
so as to include finite site variables 
(and also to the elliptic case in the former).
In Section \ref{sec:discussion} 
we will also comment on the restricted version and difficulties 
to associate integrable stochastic models.

\section{Results}\label{sec:result}
Let 
$\theta_1(u)= \theta_1(u,p)= 2p^{\frac{1}{4}}\sin\pi u\prod_{k=1}^\infty
(1-2p^{2k} \cos 2\pi u +p^{4k})(1-p^{2k})$
be one of the Jacobi theta function $(|p|<1)$ enjoying 
the quasi-periodicity
\begin{align}\label{qp}
\theta_1(u+1; e^{\pi i \tau}) = -\theta_1(u; e^{\pi i \tau}),
\qquad
\theta_1(u+\tau; e^{\pi i \tau}) = -
e^{-\pi i \tau - 2\pi i u}\theta_1(u; e^{\pi i \tau}),
\end{align}
where $\mathrm{Im} \tau > 0$.
We set
\begin{align}\label{bu}
&[u]= \theta_1(\frac{u}{L},p),\quad
[u]_k = [u][u-1]\cdots [u-k+1],
\qquad
\left[{u \atop k}\right]=\frac{[u]_k}{[k]_k} \qquad(k \in \Z_{\ge 0}),
\end{align}
with a nonzero parameter $L$.
These are elliptic analogue of the $q$-factorial and the $q$-binomial:
\begin{align*}
(z)_m = (z;q)_m = \prod_{i=0}^{m-1}(1-z q^i),\qquad
\binom{m}{l}_q= \frac{(q)_m}{(q)_l(q)_{m-l}}.
\end{align*}

For $\alpha=(\alpha_1,\ldots, \alpha_k)$ with any $k$ 
we write $|\alpha| = \alpha_1+\cdots + \alpha_k$.
The relation $\beta \ge \gamma$ or equivalently $\gamma \le \beta$ 
means $\beta_i\ge \gamma_i$ for all $i$.

We take the set of local states as
$\tilde{\mathcal{P}}= \eta+ \Z^{n+1}$
with a generic $\eta \in \C^{n+1}$.
Given positive integers $l$ and $m$,
let $a,b,c,d \in \tilde{\mathcal{P}}$ be the elements
such that  
\begin{align}\label{abcd}
\alpha = d-a \in B_l ,\quad
\beta = c-d \in B_m,\quad
\gamma= c-b \in B_l,\quad
\delta= b-a \in B_m,
\end{align}
where $B_m$ is defined by 
\begin{align}\label{bm}
B_m = \{\alpha=(\alpha_1,\ldots, \alpha_{n+1}) \in \Z_{\ge 0}^{n+1}
\mid |\alpha| = m\}.
\end{align}
The relations (\ref{abcd})  imply
$\alpha+\beta= \gamma+\delta$.
The situation is summarized as

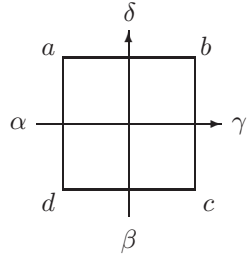
\begin{picture}(100,110)(-190,-30)

\put(0,50){\line(1,0){50}}
\put(0,0){\line(0,1){50}}
\put(50,0){\line(0,1){50}}
\put(0,0){\line(1,0){50}}

\put(25,-10){\vector(0,1){70}}
\put(-10,25){\vector(1,0){70}}

\put(-8,51){$a$}\put(52,51){$b$}
\put(-8,-8){$d$}\put(53,-8){$c$}

\put(23,64){$\delta$}
\put(-20,23){$\alpha$}\put(64,23){$\gamma$}
\put(22,-22){$\beta$}
\end{picture}

To the above configuration round a face
we assign a function 
of the spectral parameter $u$ 
called Boltzmann weight.
Its unnormalized version, denoted by
$\overline{W}_{l,m}\Bigl({a \;\;  b \atop d \;\; c}\Bigl|  u\Bigr)$,
is constructed from the $l=1$ case as follows:
\begin{align}
&\overline{W}_{l,m}\Bigl({a \;\;  b \atop d \;\; c}\Bigl|  u\Bigr)
= \sum \prod_{i=0}^{l-1}
\overline{W}_{1,m}
\Bigl({a^{(i)} \;\;\;\;\;\;  b^{(i)} \atop a^{(i+1)} \;\; b^{(i+1)}}\Bigl|  u-i\Bigr),
\label{wf}
\\
&\overline{W}_{1,m}
\Bigl({a \;\;  b \atop d \;\; c}\Bigl|  u\Bigr)
= \frac{[u+b_\nu-a_\mu]\prod_{j=1 \,(j\neq \mu)}^{n+1}
[b_\nu-a_j+1]}{\prod_{j=1}^{n+1} [c_\nu-b_j]}
\quad
(d = a + {\bf e}_\mu, \;\;  c = b + {\bf e}_\nu),
\nonumber
\end{align}
where 
${\bf e}_i = (0,\ldots, 0,\overset{i \,{\mathrm th}}{1},0,\ldots,0)$. 
In (\ref{wf}), $a^{(0)},\ldots, a^{(l)} \in \tilde{\mathcal{P}}$ 
is a path form $a^{(0)}=a$ to $a^{(l)}=d$ such that 
$a^{(i+1)}-a^{(i)} \in B_1\,(0 \le i <l)$.
The sum is taken over
$b^{(1)}, \ldots, b^{(l-1)} \in \tilde{\mathcal{P}}$ satisfying the conditions
$b^{(i+1)}-b^{(i)} \in B_1\,(0 \le i <l)$ with $b^{(0)}=b$ and $b^{(l)}=c$.
It is independent of the choice of
$a^{(1)}, \ldots, a^{(l-1)}$ (cf. \cite[Fig.2.4]{DJKMO}).
We understand that $\overline{W}_{l,m}
\Bigl({a \;\;  b \atop d \;\; c}\Bigl|  u\Bigr)=0$
unless (\ref{abcd}) is satisfied for some $\alpha,\beta, \gamma,\delta$. 

The normalized weight is defined by
\begin{align}
W_{l,m}\Bigl({a \;\;  b \atop d \;\; c}\Bigl|  u\Bigr)
&= \overline{W}_{l,m}\Bigl({a \;\;  b \atop d \;\; c}\Bigl|  u\Bigr)
\frac{[1]^l}{[l]_l}
\Bigl[{m \atop l}\Bigr]^{-1}.
\label{w89}
\end{align}
It satisfies \cite{JKMO} the (unrestricted) star-triangle relation 
(or dynamical Yang-Baxter equation) \cite{Bax}:
\begin{equation}\label{str}
\begin{split}
&\sum_g
W_{k,m}\Bigl({a \;\;  b \atop f \;\; g}\Bigl|  u\Bigr)
W_{l,m}\Bigl({f \;\;  g \atop e \;\; d}\Bigl|  v\Bigr)
W_{k,l}\Bigl({b \;\;  c \atop g \;\; d}\Bigl|  u-v\Bigr)\\
&=\sum_g
W_{k,l}\Bigl({a \;\;  g \atop f \;\; e}\Bigl|  u-v\Bigr)
W_{l,m}\Bigl({a \;\;  b \atop g \;\; c}\Bigl|  v\Bigr)
W_{k,m}\Bigl({g \;\;  c \atop e \;\; d}\Bigl|  u\Bigr),
\end{split}
\end{equation}
where the sum extends over $g \in \tilde{\mathcal{P}}$
giving nonzero weights.
Under the same setting (\ref{abcd}) as in (\ref{w89}), 
we introduce the product
\begin{align}
S_{l,m}\Bigl({a \;\;  b \atop d \;\; c}\Bigr)
= \Bigl[{m \atop l}\Bigr]^{-1}
\prod_{1 \le i, j \le n+1}
\frac{[c_i-d_j]_{c_i-b_i}}{[c_i-b_j]_{c_i-b_i}}.
\label{slm}
\end{align}
Note that $S_{l,m}\Bigl({a \;\;  b \atop d \;\; c}\Bigr) =0$
unless $d \le b$ because of
the factor $\prod_{i=1}^{n+1}[c_i-d_i]_{c_i-b_i}$.
The following result giving an explicit factorized formula of 
the weight $W_{l,m}$ at special value of the spectral parameter
is the elliptic face model analogue of 
\cite[Th.2]{KMMO}.

\begin{theorem}\label{th:s}
If $l \le m$,  
the following equality is valid:
\begin{align}\label{mina}
W_{l,m}\Bigl({a \;\;  b \atop d \;\; c}\Bigl|  u=0\Bigr)
= S_{l,m}\Bigl({a \;\;  b \atop d \;\; c}\Bigr).
\end{align}
\end{theorem}

\begin{proof}
We are to show
\begin{align}\label{kr}
\overline{W}_{l,m}\Bigl({a \;\;  b \atop d \;\; c}\Bigl|  0\Bigr)
= \frac{[l]_l}{[1]^l}
\prod_{i,j}
\frac{[c_i-d_j]_{c_i-b_i}}{[c_i-b_j]_{c_i-b_i}}.
\end{align}
Here and in what follows unless otherwise stated, 
the sums and products are taken always over $1,\ldots, n+1$
under the condition (if any) written explicitly. 
We invoke the induction on $l$.
It is straightforward to check (\ref{kr}) for $l=1$.
By the definition (\ref{wf}) the $l+1$ case is expressed as
\begin{align*}
\overline{W}_{l+1,m}\Bigl({a \;\;  b \atop d \;\; c}\Bigl|  0\Bigr)
= \sum_{\nu}
\overline{W}_{l,m}\Bigl({a \;\;\;  b \atop d' \;\; c'}\Bigl|  0\Bigr)
\overline{W}_{1,m}\Bigl({d' \,\;  c' \atop d \;\; c}\Bigl|  -l\Bigr)\qquad
(d'=d-{\bf e}_\mu, c'=c-{\bf e}_\nu)
\end{align*}
for some fixed $\mu \in [1,n+1]$.
Due to the induction hypothesis on 
$\overline{W}_{l,m}$,
the equality to be shown becomes
\begin{equation}\label{kr2}
\begin{split}
&\sum_\nu \frac{[l]_l}{[1]^l}
\Bigl(\prod_{i,j}
\frac{[c'_i-d'_j]_{c'_i-b_i}}{[c'_i-b_j]_{c'_i-b_i}}
\Bigr)
\frac{[-l+c'_\nu - d'_\mu]\prod_{k \neq \mu}[c'_\nu- d'_k+1]}
{\prod_k[c_\nu-c'_k]}
\\
&=
\frac{[l+1]_{l+1}}{[1]^{l+1}}
\prod_{i,j}
\frac{[c_i-d_j]_{c_i-b_i}}{[c_i-b_j]_{c_i-b_i}}.
\end{split}
\end{equation} 
After removing common factors using 
$c'_i = c_i - \delta_{i \nu}, d'_i = d_i - \delta_{i \mu}$,
one finds that (\ref{kr2}) is equivalent to 
\begin{align*}
\sum_\nu  [c_\nu - d_\mu -l]\prod_{i \neq \nu}
\frac{[c_i - d_\mu+1]}{[c_\nu-c_i]}
\prod_{j}[c_\nu-b_j] 
= [l+1]\prod_i [b_i-d_\mu+1]
\end{align*}
with $l$ determined by $l+1= \sum_j(c_j-b_j)$.
One can eliminate $d_\mu$ and rescale the 
variables by 
$(b_j,c_j) \rightarrow 
(Lb_j+d_\mu, Lc_j+d_\mu)$ for all $j$.
The resulting equality follows from Lemma \ref{le:id}.
\end{proof}

\begin{lemma}\label{le:id}
Let $b_1,\ldots, b_n, c_1, \ldots, c_n \in \C$ be generic and 
set $s = \sum_{i=1}^n(c_i-b_i)$.
Then for any $n \in \Z_{\ge 1}$ 
the following identity holds:
\begin{align*}
\sum_{i=1}^n\theta_1(z+c_i - s)
\prod_{j=1\,(j \neq i)}^n
\frac{\theta_1(z+c_j)}{\theta_1(c_i-c_j)}\prod_{j=1}^n\theta_1(c_i-b_j)
= \theta_1(s)\prod_{i=1}^n\theta_1(z+b_i).
\end{align*}
\end{lemma}

\begin{proof}
Denote the $\text{LHS}-\text{RHS}$ by $f(z)$.
From (\ref{qp}) we see that $f(z)$ satisfies 
(\ref{jm}) with $B=\frac{n}{2}$,
$A_1 = \frac{n(1+\tau)}{2}+ \sum_{j=1}^nb_j$ and $A_2=n$.
Moreover it is easily checked that $f(z)$ possesses zeros 
at $z = -c_1,\ldots, -c_n$.
Therefore Lemma \ref{le:jm} claims
$-(c_1+\cdots + c_n) - (B \tau + \frac{1}{2}A_2-A_1) \equiv 0$
mod $\Z+\Z\tau$. 
But this gives $s\equiv0$ which is a contradiction since
$b_j, c_j$ can be arbitrary. 
Therefore $f(z)$ must vanish identically.
\end{proof}

\begin{lemma}\label{le:jm}
Let $\mathrm {Im} \tau >0$.
Suppose an entire function $f(z) \not\equiv 0$ satisfies the quasi-periodicity 
\begin{align}\label{jm}
f(z+1) = e^{-2\pi i B}f(z),\qquad
f(z+\tau) = e^{-2\pi i (A_1+A_2z)}f(z).
\end{align}
Then $A_2 \in \Z_{\ge 0}$ holds and 
$f(z)$ has exactly $A_2$ zeros $z_1, \ldots, z_{A_2} \!\mod \Z+\Z\tau$.
Moreover 
$z_1+ \cdots + z_{A_2} \equiv 
B \tau + \frac{1}{2}A_2-A_1\,\mod \Z+ \Z\tau$ holds. 
\end{lemma}

\begin{proof}
Let $C$ be a period rectangle $(\xi, \xi+1, \xi+1+\tau, \xi+\tau)$ 
on which there is no zero of $f(z)$.
From the Cauchy theorem the number of zeros of $f(z)$ in $C$ is equal to  
$\int_C \frac{f'(z)}{f(z)}\frac{dz}{2\pi i}$.
Calculating the integral by using (\ref{jm}) one gets $A_2$.
The latter assertion can be shown similarly by considering the 
integral $\int_C \frac{zf'(z)}{f(z)}\frac{dz}{2\pi i}$.
\end{proof}

From Theorem \ref{th:s} and (\ref{str})  it follows that 
$S_{l,m}\Bigl({a \;\;  b \atop d \;\; c}\Bigr)$ also satisfies the 
(unrestricted) star-triangle relation (\ref{str}) without spectral parameter.
The discrepancy of the factorizing 
points $u=0$ in (\ref{mina}) and ``$u=l-m$" in \cite[Th.2]{KMMO} 
is merely due to a conventional difference in defining the face and the vertex weights.

Since (\ref{w89}) and (\ref{slm}) are homogeneous of degree 0 
in the symbol $[\cdots]$, the trigonometric limit 
$p \rightarrow 0$ may be understood as replacing (\ref{bu}) by   
$[u]=q^{u/2}-q^{-u/2}$ with generic $q=\exp\frac{2\pi i}{L}$.
Under this prescription the elliptic binomial 
$\left[{m \atop l}\right]$ from (\ref{bu}) is replaced by
$q^{l(l-m)/2}\binom{m}{l}_{q}$, therefore the trigonometric limit of 
(\ref{slm}) becomes
\begin{align}\label{st}
S_{l,m}\Bigl({a \;\;  b \atop d \;\; c}\Bigr)_{\mathrm{trig}}
= \binom{m}{l}_q ^{-1}\prod_{1 \le i,j \le n+1}
\frac{(q^{b_i-d_j+1})_{c_i-b_i}}
{(q^{b_i-b_j+1})_{c_i-b_i}}.
\end{align}
The following result is a trigonometric face model analogue of 
\cite[Th.6]{KMMO}.

\begin{theorem}\label{th:s1}
Suppose $l\le m$. 
Then the sum-to-1 holds in the trigonometric case:
\begin{align}
\sum_{b}S_{l,m}\Bigl({a \;\;  b \atop d \;\; c}\Bigr)_{\mathrm{trig}}= 1,
\label{s1}
\end{align}
where the sum runs over those $b$ satisfying 
$c-d \in B_m$ and $d-a \in B_l$.
\end{theorem}

\begin{proof}
The relation (\ref{s1}) is equivalent to 
\begin{align}\label{hna}
\binom{m}{l}_q = \sum_{\gamma \in B_l, \gamma \le \beta}
\prod_{1 \le i,j \le n+1}
\frac{(q^{c_{ij}-\gamma_i+\beta_j+1})_{\gamma_i}}
{(q^{c_{ij}-\gamma_i+\gamma_j+1})_{\gamma_i}} 
\qquad 
(c_{ij}= c_i-c_j)
\end{align}
for any fixed $\beta=(\beta_1,\ldots, \beta_{n+1}) \in B_m$, $l \le m$ and the 
parameters $c_1,\ldots, c_{n+1}$, 
where the sum is taken over $\gamma \in B_l$ (\ref{bm})  
under the constraint $\gamma \le \beta$.
In fact we are going to show 
\begin{align}\label{yme}
\frac{(w_1^{-1}\ldots w_n^{-1}q^{-l+1})_l}{(q)_l}
=\sum_{|\gamma|=l}\prod_{1 \le i,j \le n}
\frac{\bigl(q^{-\gamma_i+1}z_i/(z_jw_j)\bigr)_{\gamma_i}}
{\bigl(q^{\gamma_j-\gamma_i+1}z_i/z_j\bigr)_{\gamma_i}}
\qquad (l \in \Z_{\ge 0}),
\end{align}
where the sum is over $\gamma \in \Z_{\ge 0}^{n}$
such that $|\gamma|=l$, and 
$w_1,\ldots, w_n, z_1,\ldots, z_n$ are arbitrary parameters.
The relation (\ref{hna}) is deduced from 
$(\ref{yme})|_{n\rightarrow n+1}$ 
by setting $z_i = q^{c_i}, w_i = q^{-\beta_i}$ 
and specializing $\beta_i$'s to  
nonnegative integers.
In particular, the constraint $\gamma \le \beta$ automatically arises from the $i=j$ 
factor $\prod_{i=1}^n(q^{-\gamma_i+1+\beta_i})_{\gamma_i}$ 
in the numerator.
To show (\ref{yme}) we rewrite it slightly as
\begin{align}\label{sw}
q^{\frac{l^2}{2}}\frac{(w_1\ldots w_n)_l}{(q)_l}
= \sum_{|\gamma| = l}\prod_{i=1}^n 
q^{\frac{\gamma_i^2}{2}}\frac{(w_i)_{\gamma_i}}{(q)_{\gamma_i}}
\prod_{1 \le i\neq j \le n}
\frac{(z_jw_j/z_i)_{\gamma_i}}{(q^{-\gamma_j}z_j/z_i)_{\gamma_i}}.
\end{align} 
Denote the RHS by $F_n(w_1, \ldots, w_n|z_1, \ldots, z_n)$.
We will suppress a part of the arguments when they are kept unchanged
in the formulas.
It is easy to see 
\[
F_n(w_1,w_2|z_1,z_2)=
F(w_2,w_1|z_2,z_1)
=F_n(\frac{z_2w_2}{z_1}, \frac{z_1w_1}{z_2}|z_1,z_2).
\]
Thus the coefficients in the expansion
$F_n(w_1,w_2|z_1, z_2)= \sum_{0 \le i,j \le l}C_{i,j}(z_1,z_2)w_1^iw_2^j$ 
are rational functions in $z_1,\ldots, z_n$ obeying
$C_{i,j}(z_1,z_2) = C_{j,i}(z_2,z_1) 
= \bigl(\frac{z_1}{z_2}\big)^{i-j}C_{j,i}(z_1,z_2)$.
On the other hand from the explicit formula (\ref{sw}), 
one also finds that 
any $C_{i,j}(z_1,z_2)$ remains finite in the either limit
$\frac{z_1}{z_2}, \frac{z_2}{z_i}\rightarrow \infty$
or  
$\frac{z_1}{z_2}, \frac{z_2}{z_i}\rightarrow 0$ 
for $i\ge 3$.
It follows that $C_{i,j}(z_1,z_2)=0$ unless $i=j$, hence 
$F_n(w_1, w_2,\ldots, w_n |z_1,\ldots, z_n)
=F_n(1,w_1w_2, w_3,\ldots, w_n |z_1, \ldots, z_n )$.
Moreover it is easily seen 
$F_n(1,w_1w_2, w_3,\ldots, w_n |z_1,z_2,\ldots, z_n )
=F_{n-1}(w_1w_2, w_3,\ldots, w_n |z_2,\ldots, z_n )$.
Repeating this we reach $F_1(w_1\cdots w_n|z_n)$ giving the 
LHS of (\ref{sw}).
\end{proof}

We note that the sum-to-1 (\ref{s1}) does not hold in the elliptic case.
Remember that our local states are taken from 
$\tilde{\mathcal{P}} = \eta + \Z^{n+1}$ with a generic $\eta \in \C^{n+1}$.
So we set $a= \eta + {\tilde a}$ with ${\tilde a} \in \Z^{n+1}$ etc
in (\ref{bm}), and assume that it is valid also for 
${\tilde a}, {\tilde b}, {\tilde c}, {\tilde d}$.
It is easy to check
\begin{proposition}\label{pr:knh}
Assume $l \le m$ and $|q|<1$.
Then the following equality holds:
\begin{align}\label{konoha}
\lim_{\eta \rightarrow \infty}S_{l,m}
\Bigl({\eta + {\tilde a} \;\;  \eta + {\tilde b} 
\atop \eta + {\tilde d} \;\; \eta + {\tilde c}}\Bigr)_{\mathrm{trig}}
= q^{\sum_{i<j}(\beta_i-\gamma_i)\gamma_j}
\binom{m}{l}_q^{-1}\prod_{i=1}^{n+1}
\binom{\beta_i}{\gamma_i}_q,
\end{align}
where the limit means $\eta_i-\eta_{i+1} \rightarrow \infty$ 
for all $1 \le i \le n$, and the RHS is zero unless 
$0 \le \gamma_i \le \beta_i,\,\forall i$.
\end{proposition}
The limit reduces the unrestricted trigonometric 
$A^{(1)}_n$ face model to the vertex model 
at a special value of the spectral parameter in the sense that
the RHS of $(\ref{konoha})|_{q\rightarrow q^2}$ reproduces 
\cite[eq.(23)]{KMMO} that was obtained as the special value of the 
quantum $R$ matrix associated with the symmetric tensor representation
of $U_q(A^{(1)}_n)$.

\section{Discussion}\label{sec:discussion}
Since the weights $W_{l,m}\Bigl({a \;\;  b \atop d \;\; c}\Bigl|  u\Bigr)$ 
remain unchanged by shifting $a,b,c,d \in \tilde{\mathcal{P}}$ 
by $\mathrm{const}\cdot(1,\ldots, 1)$,
we regard them as elements from 
$\mathcal{P}:=\tilde{\mathcal{P}}/\C(1,\ldots,1)$ in the sequel.
Given $l, m_1,\ldots, m_M \in \Z_{\ge 1}$ and 
$u,w_1,\ldots, w_M \in \C$, the transfer matrix $T_l(u)
=T_l\left(u\left|{m_1,\ldots, m_M  \atop w_1,\ldots, w_M}\right.\right)$
of the unrestricted $A^{(1)}_n$ face model with  
periodic boundary condition is a linear map on the space of independent 
row configurations on length $M$ row
$\bigoplus \C|a^{(1)},\ldots a^{(M)}\rangle$ where the 
sum is taken over $a^{(1)},\ldots a^{(M)} \in \mathcal{P}$ such that 
$a^{(i+1)}-a^{(i)} \in B_{m_i}\,(a^{(M+1)}=a^{(1)})$.
Its action is specified as 
$T_l(u)|b^{(1)},\ldots b^{(M)}\rangle 
= \sum_{a^{(1)},\ldots a^{(M)}}
T_l(u)_{b^{(1)},\ldots b^{(M)}}^{a^{(1)},\ldots a^{(M)}}
|a^{(1)},\ldots a^{(M)}\rangle$ in terms of the matrix elements
\begin{align}
T_l(u)_{b^{(1)},\ldots b^{(M)}}^{a^{(1)},\ldots a^{(M)}}
= \prod_{i=1}^M
W_{l,m_i}\Bigl({a^{(i)} \;\;  a^{(i+1)} 
\atop b^{(i)} \;\; b^{(i+1)}}\Bigl|  u-w_i\Bigr)
\qquad(a^{(M+1)}=a^{(1)}, b^{(M+1)}=b^{(1)}).
\end{align}
Theorem \ref{th:s} tells that 
$S_l :=T_l(u)_{u=w_1=\cdots = w_M}$
has a simple factorized matrix elements.
We write its elements as
$S_{l,b^{(1)},\ldots b^{(M)}}^{\phantom{l} a^{(1)},\ldots a^{(M)}}$.
The star-triangle relation (\ref{str}) implies the commutativity 
$[T_l(u), T_{l'}(u')]=[S_l, S_{l'}]=0$.

Let us consider whether 
$X=T_l(u)$ or $S_l$ admits an interpretation as a Markov matrix of a 
discrete time stochastic process.
The related issue was treated in \cite{Bo} for $n=1$ and mainly 
when $\min(l,m_1,\ldots, m_M)=1$.
One needs (i) sum-to-1 property  
$\sum_{a^{(1)},\ldots a^{(M)}}
X_{b^{(1)},\ldots b^{(M)}}^{a^{(1)},\ldots a^{(M)}}=1$
and (ii) nonnegativity 
$\forall X_{b^{(1)},\ldots b^{(M)}}^{a^{(1)},\ldots a^{(M)}}\ge 0$.
We concentrate on the trigonometric case in what follows.
From Theorem \ref{th:s} and the fact that 
$S_{l,m}\Bigl({a \;\; b \atop c \;\; d}\Bigr)_{\mathrm{trig}}$
in (\ref{st}) is 
independent of $a$, (i) indeed holds for $S_l$.
On the other hand (\ref{st}) also indicates that  
(ii) is not valid in general without confining the 
site variables in a certain range.
A typical such prescription is 
{\em restriction} \cite{DJKMO, JMO, JKMO}, where  
one takes $L=\ell+n+1$ in (\ref{bu}) with some $\ell\in \Z_{\ge 1}$
and lets the site variables range over the finite set of 
level $\ell$ dominant integral weights   
$\{(L+a_{n+1}-a_1-1)\Lambda_0 
+ \sum_{i=1}^n(a_i-a_{i+1}-1)\Lambda_i
\mid L+a_{n+1}> a_1> \cdots > a_{n+1}, a_i-a_j \in \Z \}$.
They are to obey a stronger adjacency condition \cite[p546, (c-2)]{JKMO}
than (\ref{abcd}) which is actually the fusion rule of the WZW conformal
field theory. 
(The formal limit $\ell \rightarrow \infty$ still works to
restrict the site variables to 
the positive Weyl chamber
and is called ``classically restricted".)
Then the star-triangle relation remains valid by virtue of nontrivial cancellation of unwanted terms.
However, discarding the contribution to the sum (\ref{s1})
from those $b$ not satisfying the adjacency condition
spoils the sum-to-1 property.
For example when $(n,l,m)=(2,1,2), a=(2,1,0), c=(4,2,0), d=(3,1,0)$ and 
$\ell$ is sufficiently large,
the unrestricted sum (\ref{s1}) consists of two terms 
$S_{l,m}\Bigl({a \;\;  b \atop d \;\; c}\Bigr)_{\mathrm{trig}}
=\binom{2}{1}_q^{-1}\frac{(q^{-1};q)_1}{(q^{-2};q)_1}$ for $b=(4,1,0)$
and
$S_{l,m}\Bigl({a \;\;  b' \atop d \;\; c}\Bigr)_{\mathrm{trig}}
=\binom{2}{1}_q^{-1}\frac{(q^{3};q)_1}{(q^{2};q)_1}$ for $b'=(3,2,0)$
summing up to 1, but 
$b'$ must be discarded in the restricted case
since $a \overset{m=2}{\Rightarrow} b'$ \cite[(c-2)]{JKMO}
does not hold.
Thus we see that in order to satisfy 
(i) and (ii) simultaneously one needs to resort to a construction
different from the restriction.

\section*{Acknowledgments}
The author thanks Masato Okado for discussion.
This work is supported by 
Grants-in-Aid for Scientific Research 
No.~15K13429 from JSPS.


\begin{thebibliography}{99}

\bibitem {ABF}
G.~E.~Andrews, R.~J.~Baxter and P.~J.~Forrester,
Eight vertex SOS model and generalized
Rogers-Ramanujan-type identities,
J. Stat. Phys. {\bf 35}, (1984) 193--266.

\bibitem{Bax}
R.~J.~Baxter,
{\it Exactly solved models in statistical mechanics},
Academic Press, London (1982).

\bibitem{Bo}
A.~Borodin,
Symmetric elliptic functions, IRF models, 
and dynamic exclusion processes,
arXiv:1701.05239.

\bibitem{DJKMO}
E.~Date, M.~Jimbo, A.~Kuniba, T.~Miwa and M.~Okado,
Exactly solvable SOS models II:
Proof of the star-triangle relation and combinatorial identites,
Adv. Stud. in Pure Math. {\bf 16} (1988) 17-122.

\bibitem{JKMO}
M.~Jimbo, A.~Kuniba, T.~Miwa and M.~Okado,
The $A^{(1)}_n$ face models, 
Commun. Math. Phys. {\bf 119} (1989) 543--565.

\bibitem{JMO}
M.~Jimbo, T.~Miwa and M.~Okado,
Symmetric tensors of the $A^{(1)}_{n-1}$ family,
Algebraic Analysis, {\bf 1} (1988) 253--266.

\bibitem{Kuan}
J.~Kuan,
An algebraic construction of duality functions for the stochastic 
$U_q(A_n^{(1)})$ vertex model and its degenerations,
arXiv:1701.04468.

\bibitem{KMMO}
A.~Kuniba, V.~V.~Mangazeev, S.~Maruyama and M.~Okado,
Stochastic $R$ matrix for $U_q(A^{(1)}_n)$,
Nucl. Phys. B{\bf 913} (2016) 248--277.


\end{thebibliography}
\end{document}